%% file: Efficient_Private_Distributed_Computation_on_Unbounded_Input_Streams.tex
\documentclass[letterpaper,11pt]{article}

\usepackage{fullpage}
\usepackage{amsmath,amsthm}
\usepackage{accents}
\usepackage{times}
\usepackage{url}
\usepackage{float}
\usepackage{wrapfig}
\usepackage{cite}
\usepackage{mdwlist}
\usepackage{graphicx,amsmath,amssymb,amsthm,subfigure, amsfonts}
\usepackage[ruled,vlined]{algorithm2e}
\usepackage{algorithmic}
\usepackage{tipa}
\usepackage{a4wide}
\usepackage{color}

\usepackage{subfigure}
\usepackage[english]{babel}
\usepackage{subfigure}
\usepackage{nonfloat}
\usepackage{epsfig}
\usepackage{sidecap}
\usepackage{authblk}

\usepackage{algorithm2e}

\newcommand{\ignore}[1]{}
               {}
\newcommand{\abs}[1]{\left| #1\right|}

\newcommand{\Ad}{$Adv$}
\newcommand{\A}{$\cal A$}
\newcommand{\rmPi}{\rm {\Pi}}
\newcommand{\nnPi}{{\rm {\Pi}}^{(n,n)}}
\newcommand{\tnPiNaive}{{\rm {\Pi}}^{(t+1,n)}_{naive}}
\newcommand{\tnPi}{{\rm {\Pi}}^{(t+1,n)}}
\newcommand{\view}{\mathrm{VIEW}}
\newcommand{\seed}{\mathit{seed}}
\newcommand{\init}{\mathit{init}}
\newcommand{\curr}{\mathit{curr}}
\newcommand{\next}{\mathit{next}}
\newcommand{\len}{\mathit{len}}

\newcommand{\todo}[1]{{\em{TODO: {#1}}}}

\newcommand{\compind}{\ensuremath{\stackrel{\mathsf{c}}{\mathop\approx}}}

\newtheorem{definition}{Definition}
\newtheorem{theorem}{Theorem}[section]
\newtheorem{lemma}[theorem]{Lemma}

\newtheorem{proposition}[theorem]{Proposition}

\newtheorem{notation}{Notation}

\newtheorem{remark}[theorem]{Remark}

\newenvironment{proofsk}{\begin{proof}[Proof sketch.]}
{\end{proof}}

\newlength{\saveparindent}
\newlength{\saveparskip}
\newcounter{ctr}

\newenvironment{tiret}{%
\begin{list}{\hspace{1pt}\rule[0.5ex]{6pt}{1pt}\hfill}{\labelwidth=15pt%
\labelsep=3pt \leftmargin=18pt \topsep=1pt%
\setlength{\listparindent}{\saveparindent}%
\setlength{\parsep}{\saveparskip}%
\setlength{\itemsep}{1pt}}}{\end{list}}

\newenvironment{newenum}{%
\begin{list}{{\rm \arabic{ctr}.}\hfill}{\usecounter{ctr}\labelwidth=17pt%
\labelsep=6pt \leftmargin=23pt \topsep=.5pt%
\setlength{\listparindent}{\saveparindent}%
\setlength{\parsep}{\saveparskip}%
\setlength{\itemsep}{5pt} }}{\end{list}}

\begin{document}

\begin{titlepage}

\title{\bf Efficient Private Distributed Computation \\
on Unbounded Input Streams\thanks{This research has been supported by the Israeli Ministry of Science and Technology (MOST), the Institute for Future Defense Technologies Research named for the Medvedi, Shwartzman and Gensler Families, the Israel Internet Association (ISOC-IL), the Lynne and William Frankel Center for Computer Science at Ben-Gurion University, Rita Altura Trust Chair in Computer Science, {\em Israel Science Foundation} (grant number 428/11), Cabarnit Cyber Security MAGNET Consortium, MAFAT and Deutsche Telekom Labs at BGU. Emails: {\tt {dolev,yuditsky}\allowbreak @cs.bgu.ac.il}, {\tt garay@\allowbreak research.\allowbreak att.com}, {\tt niv.gilboa@gmail.com}, {\tt kolesnikov@\allowbreak research.bell-labs.com.} A brief announcement will be presented in DISC 2012.}}

\author[1]{Shlomi Dolev}
\author[2]{Juan Garay}
\author[3]{Niv Gilboa}
\author[4]{Vladimir Kolesnikov}
\author[1]{Yelena Yuditsky}

\affil[1]{Department of Computer Science, Ben Gurion University of the Negev, Israel }
\affil[2]{AT\&T Labs -- Research, Florham Park, NJ }
\affil[3]{Deptartment of Communication Systems Engineering, Ben-Gurion University of the Negev, Beer-Sheva, Israel}
\affil[4]{Bell Laboratories, Murray Hill, NJ}

%


\date{}

\maketitle
\thispagestyle{empty}
\begin{abstract}
In the problem of swarm computing, $n$ agents wish to securely and
distributively perform a computation on common inputs, in such a way
that even if the entire memory contents of some of them are exposed,
no information is revealed about the state of the computation.
Recently, Dolev, Garay, Gilboa and Kolesnikov [ICS 2011] considered this
problem in the setting of information-theoretic security, showing how to perform such computations on input streams of {\em unbounded length}.  The cost of their solution, however, is exponential in the size of the Finite State Automaton (FSA) computing the function.

In this work we are interested in efficient computation in the above
model, at the expense of {\em minimal} additional assumptions. Relying on the existence of one-way functions, we show how to process {\it a priori} unbounded inputs (but of course, polynomial in the security parameter) at a cost {\em linear} in $m$, the number of FSA states. In particular, our algorithms
achieve the following: 

\begin{tiret}

\item In the case of $(n,n)$-reconstruction (i.e. in which all $n$
agents participate in reconstruction of the distributed computation) and at most $n-1$ agents are corrupted,
the agent storage, the time required to process each input symbol and the time complexity for reconstruction are all $O(mn)$. 

\item In the case of $(t+1,n)$-reconstruction (where only $t+1$ agents 
take part in the reconstruction) and at most $t$ agents are corrupted,
the agents' storage and time required to process each input symbol are $O(m{n-1 \choose t-1})$. The complexity of reconstruction is $O(m(t+1))$.

\end{tiret}
\end{abstract}


\end{titlepage}
\clearpage
\pagenumbering{arabic} 


\input{introduction-3}

\input{model}

\input{overview}

\section{The Constructions in Detail}
\label{sec-from-all}

\subsection{The $(n,n)$-reconstruction protocol \label{sec-from-all-details}}
We start our formalization of the intuition presented above with the
case where all $n$ out of the $n$ agents participate in the state
reconstruction.  The protocol for this case, which we call $\nnPi$, is presented
below.

\vspace{-.1in}
\paragraph{Protocol $\nnPi$.} The protocol consists of three phases:

\vspace{.1in}

\noindent {\em Initialization.} The dealer secret-shares among the agents 
a secret value for each state, such that the value for the initial
state is $1$ and for all the other states is $0$. This is done as follows. 
Agent $A_i$ ($1\leq i\leq n$) is given a a random binary string
$x^i_1x^i_2...x^i_m$, with the constraints that 
$$x^1_{\init} + x^2_{\init} + ... + x^n_{\init}\equiv 1 \bmod 2,$$ 
where $\init$ is the index of the initial state of the
computation, and for every $1\leq j \neq init \leq m$, 

$$x^1_j + x^2_j + ... + x^n_j \equiv 0 \bmod 2.$$ 
Each agent then proceeds to assign its state labels as
$\ell_j^i \leftarrow x_j^i$.

\vspace{.1in}
\noindent {\em Event Processing.}  Each agent runs Algorithm~\ref{algo:calc},
updating its labels and computing the new seeds for the PRG.  
Let ${\cal T}$ be the set of all possible agents' pairs. For line 8
of Algorithm \ref{algo:calc}, each agent $A_i$ now computes
$$
R_j= \sum_{\substack{T\in {\cal T}, A_i\in T}} (b_j^T)_r.
$$


\begin{figure*} [ht]
\begin{center}
\includegraphics[width=7cm]{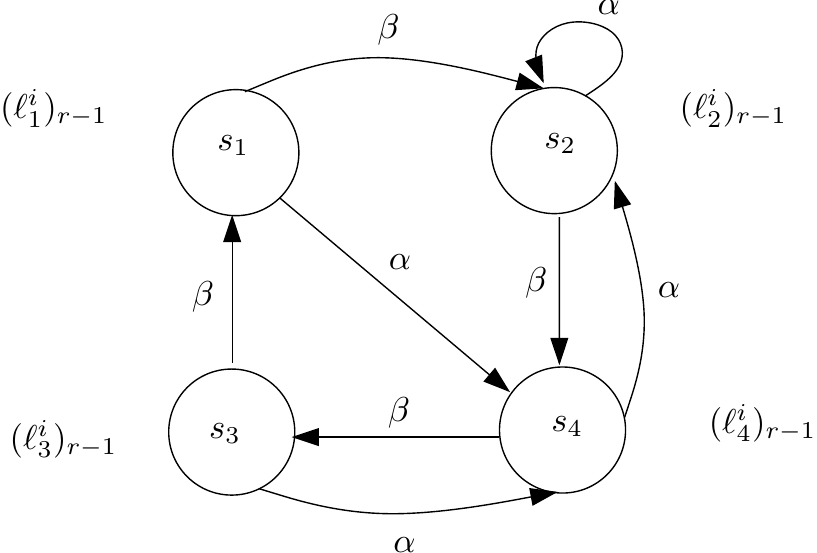}
\caption{\sl The internal state of agent $A_i$ before a transition.}
\end{center}
\label{fig:comp_scheme1}
\end{figure*}

\vspace{.1in}
\noindent {\em Reconstruction.}  All agents submit their internal states to the
dealer, who reconstructs the secrets  corresponding to each state, by 
adding (mod 2) the shares of each state, and determines and outputs 
the currently active state (the one whose reconstructed secret is $1$).

\begin{figure*} [ht]
\begin{center}
\includegraphics[width=11cm]{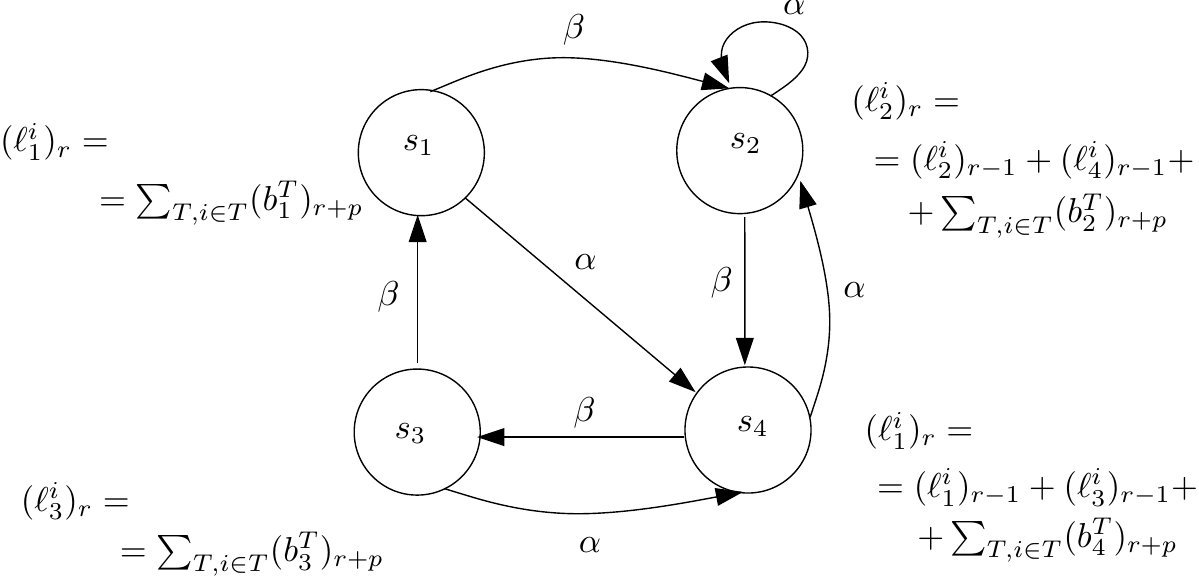}
\caption{\sl The internal state of agent $A_i$ after an $\alpha$ transition.}
\end{center}
\label{fig:comp_scheme2}
\end{figure*}

\vspace{.1in}
Before proving the correctness and privacy achieved by the protocol, 
we illustrate the operation of the online (Event Processing) phase
with the following example;
refer to Figures 1
and 2.
The two figures describe the execution of the
protocol
on an automaton with four states and two possible
inputs. 
Figure 1 presents the 
internal state of agent $A_i$ after the $(r-1)$-th clock cycle. 
The agent holds the original automaton and has a label for each of 
the four states, 
$(\ell^i_1)_{r-1}$, $(\ell^i_2)_{r-1}$, $(\ell^i_3)_{r-1}$ and
$(\ell^i_4)_{r-1}$.

Figure 2
shows the changes in the agent's internal state compared to Figure 1
after the $r$-th clock cycle. We also assume
that in this clock cycle the agents receive an input symbol
$\alpha$. The new labels for each state are the sum of old labels and
pseudo-random values. The labels in the sum are the old labels of
all the states that transition to the current state given the
input. Thus, the new $(\ell_2^i)_r$ includes a sum of the old
$(\ell^i_2)_{r-1}$ and the old $(\ell^i_4)_{r-1}$, while the new
$(\ell_3^i)_r$ doesn't include any labels in its sum because there is
no state that transitions to $s_3$ after an $\alpha$ input. The 
pseudo-random addition to each state $j=1,\ldots,4$ is the sum $\sum_{T, i
\in T} (b_j^T)_{r}$.

We start by proving  the correctness of the construction.


\begin{proposition}
\label{prop-all}
At every Event Processing step of protocol$\nnPi$'s,
the secret corresponding
to the current state in the computation is $1$ and for all other states 
the secret is $0$.
\end{proposition}
\begin {proof}
The proof is
by induction on the number of steps $r$ that the automaton performs, i.e., 
the number of clock cycles.

For the base case, if we consider the state of the protocol
after the initialization step and before the first clock cycle, i.e.,
at $r=0$, then the statement is true by our definition of the label
assignments. Let us first consider the case where at the $r$-th step
an input symbol $\gamma _r$ from $\Gamma$ is received.  
Following the protocol,
agent $A_i$'s new label for state $j$ becomes
$$ \ell ^i_j \longleftarrow
\sum_{\substack{k~ :\\ \mu(s_k, \gamma_r)=s_j}} \ell^i_k + \sum _{A_i \in
T} (b_j^T)_r.
$$

Consider now the next state of the computation in the automaton; we
wish to show that the secret corresponding to that state will be 1.
Let $\curr$ be the index of the current state of the automaton, and
$\next$ be the index corresponding to the next state; by definition,
$\mu(s_{\curr},\gamma_r)=s_{\next}$. Then, 
$$\ell^i_{\next} \longleftarrow
\sum_{\substack{k~:\\ \mu(s_k, \gamma_r)=s_{\next}}} \ell^i_k = \\
$$ 
$$
\ell^i_{\curr}+ \sum_{\substack{k\neq \curr~:\\ \mu(s_k,
\gamma_r)=s_{\next}}} \ell^i_k + \sum _{i \in T} (b_j^T)_r.$$ 

By the induction hypothesis, we know that 
$$\sum_{i=0}^n \ell^i_{\curr} \equiv 1
\mbox{ (mod 2)}$$ 
and for $k\neq \curr$, 
$$\sum_{i=0}^n \ell^i_k \equiv 0 \mbox{ (mod 2)}.$$ 

\noindent Thus, if we will sum over all the agents:

\begin{eqnarray*}
\lefteqn{ \sum_{i=0}^n \left( \ell^i_{\curr}+  \sum_{\substack{k\neq \curr~:\\ \mu(s_k, \gamma_r)=s_{\next}}} \ell^i_k + \sum _{i \in T} (b_j^T)_r \right)} \\
&=& \sum_{i=0}^n  \ell^i_{\curr} + \sum_{\substack{k \neq \curr~:\\ \mu(s_k, \gamma_r)=s_{\next}}} \sum_{i=0}^n \ell^i_k  \\
&+& \sum_{i=0}^n \sum _{i \in T} (b_j^T)_r  \equiv 1 + 0  \equiv 1 \mbox{ (mod 2)}.
\end{eqnarray*}

\noindent This is because in $ \sum_{i=1}^n \sum _{i \in T} (b_j^T)_r$, 
every $(b_j^T)_r$ appears exactly twice in this sum, once for
every 
element in $T$. 
Using similar arguments one can see that all the other
states will resolve to 0.  

In the case that in the $r$-th step no input symbol is received,
due to the fact that we just add the
random strings in the same way as in the case above,
we again get
that the secret corresponding to the current
state of the computation is 1, and for all others is 0.
\end {proof}

\begin{proposition}
\label{prop-pcm-all}
Protocol 
$\nnPi$ is $(n-1)$-private in the PCM model according to 
Definition~\ref{defProgressive}.
\end{proposition}
\begin {proofsk}
Recall that the underlying observation is that when a corruption takes place
(which cannot happen during the label-update procedure), the agent's state
includes the current labels and PRG seeds which have already been evolved, 
and hence cannot be correlated with the label shares previously
generated.

Without loss of generality, consider the case where \Ad\ corrupts
all but one agent according to an arbitrary corruption timeline,
and assume, say, agent $A_1$ is not corrupted.  We argue that the 
view of the adversary is indistinguishable from a 
view corresponding to
(randomly) initialized agents $A_2,...,A_n$ on the given automaton and
any initial state.  In other words, the view of the adversary is
indistinguishable from the view he would obtain if he corrupted the
agents simultaneously and before any input was processed.  Once we
prove that, the proposition follows.

The view of each corrupted agent includes $n-1$ seeds that he shares
with other agents and the FSA labels which are secret shares of $0$ or
a $1$.  We argue that, from the point of view of the adversary,
these labels are {\em random} shares 
of either $0$ or $1$.  This follows
from the PRG property that an evolved seed cannot be correlated with
a prior output of the PRG, and from the fact that $A_1$ remains uncorrupted.
Indeed, the newly generated ``empty'' states' labels look random since
the adversary cannot link them to the PRG seeds in his view.  The other
states' labels look random to the adversary since they are XORed
with $A_1$'s label.

Thus, the total view of the adversary consists of random shares of
$0$ and $1$, and is hence indistinguishable from the one corresponding
to the initial state.
\end{proofsk}

We now calculate time and storage complexity of $\nnPi$.
At every step of the
computation, each agent pseudo-randomly generates and XORs $n-1$ strings.
Further, each agent holds a small constant-length label for each
automaton state, and $n-1$ PRG seeds, yielding an
$O(m+n)$
memory requirement.


\subsection{The $(t+1,n)$-reconstruction protocol}
\label{sec-from-t}
Recall that in this case, up to $t$ of the agents might not take part
in the reconstruction, and thus $n>2t$.

A straightforward (albeit costly) solution to this scenario would be to execute $\nnPi$ independently for every subset of agents of size $t+1$.  This would involve each agent $A_i$ holding ${n-1 \choose t}$ copies of the automaton ${\cal A}$, one copy for each such subset which includes $A_i$, and updating them all, as in $\nnPi$, according to the same input symbol.  
Now, during the reconstruction, the dealer can recover the
output from any subset of $t+1$ agents. The cost of this approach would be as follows.  Every agent holds ${n-1 \choose t}$ automata (one for every $t+1$ tuple that includes this agent), and executes $\nnPi$, which requires $O(m+t)$ memory, resulting in a total cost of $O \big( {n-1\choose t}\cdot (m+t) \big)$,
with the cost of computation per input symbol being proportional to storage's.
 

We now present $\tnPi$, an improved $(t+1,n)$ reconstruction scheme,
whose intuition was already presented in Section~\ref{outline}. The protocol
uses Shamir's secret-sharing scheme~\cite{Sha79}, which we now 
briefly review. Let $\mathbb{F}$ be a field of
size greater than $n$,
and $s\in \mathbb{F}$ be the secret.
The dealer randomly generates coefficients $c_1,c_2,...,c_t$ from
$\mathbb{F}$ and construct the following polynomial of degree $t$,
$f(x)=s+c_1x+c_2x^2+...+c_{t}x^{t}$. 
The dealer gives each participant $A_i$, $1 \leq i \leq n$, the value $f(i)$.  It can be easily seen that one can reconstruct the
secret from any subset of at least $t+1$ points, and no information
about the secret is revealed by 
$t$ points (or less).

\vspace{-.15in}
\paragraph{Protocol $\tnPi$.} As before, the protocol consists of
three phases:

\vspace{.05in}
\noindent {\em Initialization.}  Using Shamir's secret sharing as
described above, the dealer shares a secret $1$ for the 
initial state and $0$ for all other states.
In addition, the dealer generates a random seed for every
set of $n - (t-1) = n-t+1$ agents, and gives each agent the seeds for the sets it
belongs to. 
Let ${\cal T}$ be the set of all possible subsets of
$n-t+1$ agents.

\vspace{.06in}
\noindent {\em Event Processing.} Each agent runs Algorithm~\ref{algo:calc} updating its labels, as follows.   

Let $T\in {\cal T}$ and $j$, $1 \le j\le m$, be a 
state of the automaton.  Upon obtaining value $b^T_j$ (refer to Algorithm~\ref{algo:calc}),
the agents in $T$ (individually) construct a
degree-$t$ polynomial, $P^T_j$, by defining its value on the following
$t+1$ field points: $0$, all the points $i$ such that $A_i \not \in T$,
and $k$ such that $k$ is the minimal agent's index in $T$ (the choice of 
which point in $T$ is arbitrary). Now define $P^T_j(0)=0$, 
$P^T_j(i)=0 ~ \forall A_i \not \in T$, and $P^T_j(k)=b^T_j$.

Observe that by this definition, {\em every} agent $A_i\in T$ can use
polynomial interpolation to compute $P^T_j(i)$, since the only required
information is $b^T_j$ (and the knowledge of set membership).

Let polynomial $P_j$ be defined as
$P_j=\sum_{T \in {\cal T}} P^T_j$. Each agent $A_i$ now computes 
$P_j(i)$ (note that this is possible since the values corresponding to
sets the agent does not belong to is set to $0$), and updates the 
$j$-th label, $1 \leq j \leq m$, in Algorithm \ref{algo:calc} by 
setting $R_j=P_j(i)$ in line $8$.

\vspace{.06in}
\noindent {\em Reconstruction.} At least $t+1$ agents submit their 
internal state to the dealer, who, for every $j=1,\ldots,m$, views the
$j$-th labels of $t+1$ agents as shares in a Shamir secret-sharing
scheme. The dealer reconstructs all the $m$ secrets using 
the scheme's reconstruction procedure, and determines and outputs the
currently active state (whose recovered secret is equal to $1$).


\begin{proposition}
\label{prop-threshold}
At every Event Processing step of protocol $\tnPi$, the shared secret
for the current state in the computation is $1$ and for all the other
(inactive) states, the shared secret is $0$. Furthermore, $t+1$ agents
can jointly reconstruct all secrets.
\end{proposition}

\begin {proof}
We prove the proposition by induction on the number of clock cycles $r$. We show that at each clock cycle $r$, for every state $s_j$, the $n$ labels $\ell^1_j,\ldots,\ell^n_j$ are points on a degree $t$ polynomial $Q_j$ whose free coefficient is $1$ if $j$ is the current state and $0$ otherwise.

At initialization, the 
claim is true by our definition of the label assignments. 

Assume that the induction hypothesis is correct after $r-1$. We prove
the hypothesis for the $r$-th step. Assume first that in this step the
agents receive an input letter $\gamma_r$, and denote the current
state by $s_{\curr}$. By our definition, the new label of the state $j$
of agent $i$ is

$$ \ell ^i_j \longleftarrow \sum_{\substack{k~:\\ \mu(s_k, \gamma_r)=s_j}} \ell^i_k + P_j(i),$$
\noindent or, equivalently,
$$ \ell ^i_j \longleftarrow \sum_{\substack{k~:\\ \mu(s_k, \gamma_r)=s_j}} Q_k(i) + P_j(i).$$

For every $j, 1 \leq j \leq m$, define polynomial $Q'_j$ as
$$Q'_j=\sum_{\substack{k~:\\ \mu(s_k, \gamma_r)=s_j}} Q_k + P_j.$$
Therefore, $Q'_j(i)=\ell ^i_j$ for every $j$ and every $i$. In
addition, since every $Q_k$ is of degree $t$ and so is $P_j$, we deduce
that $Q'_j$ is also of degree $t$. We finish proving the induction step by
showing that $Q'_j(0)=1$ only for the correct state.

Let $\mu(s_{\curr},\gamma_r)=s_{\next}$. By induction, $Q_{curr}(0)=1$
and $Q_j(0)=0$ for any $j \neq \curr$. Furthermore, by construction
$P_j(0)=0$, and therefore $Q'_{curr}(0)=1$. Since $Q_j(0)=0$ for any
$j \neq \curr$, we have that $Q'_j(0)=0$ for any $j \neq \next$.

If the agents do not receive any input symbol in the $r$-th clock cycle, 
then 
the claim follows by similar arguments as above.
\end {proof}

\begin{proposition}
\label{prop-pcm-tn}
$\tnPi$ is $t$-private in the PCM model according to Definition~\ref{defProgressive}.
\end{proposition}
\noindent At a high level, the proof follows the steps of the proof of 
Proposition~\ref{prop-pcm-all}. The full details of the privacy
analysis are presented in Section \ref{sec-priv-ana}.

We now calculate the costs incurred by the protocol.  The space
complexity of each agent is as follows. An agent holds a label for
every state, i.e. $m\cdot (\lceil log |\mathbb{F}| \rceil +1)$
bits. Additionally every agent holds $ {n-1 \choose n-t}={n-1 \choose \ t-1}$ seeds, where every seed is of size $\len$. Thus, in total we have ${n-1 \choose \ t-1}\cdot \len + m\cdot (\lceil log |\mathbb{F}| \rceil +1)$
bits. Each step of the Event Processing phase requires $O(m {n-1 \choose \ t-1})$ time for seed manipulation and field
operations. Reconstruction (by the dealer)
is just interpolation of $m$ polynomials of degree $t$.

\newpage
\section*{Supplementary Material:}
\section{Privacy Analysis in Detail \label{sec-priv-ana}}

\ignore{
\todo{I started reading the proofs but got confused a little.  I modified and moved the definitions part to prelims.  I wrote my own proofs in previous sections.  I am not 100\% convinced by my proofs, so please adjust and/or check if they make sense. I think the rest of this section can go.

 Also, you guys pay a lot of attention to the PRG syntax (and other minute details).  Usually in crypto papers it is taken for granted, and people don't write it all out.}

} 

We show that each of our schemes in Sections
\ref{sec-from-all-details} and \ref{sec-from-t} is computationally
private in the PCM in two stages. In the first stage we construct for
each scheme $\Pi$ and every possible corruption timeline $\rho$ an
intermediate scheme, $I(\Pi, \rho)$. We prove that if the corruption
timeline is $\rho$ then the view of an adversary in $I(\Pi, \rho)$ is
independent of the state of the automaton. In other words, the
adversary's view is distributed identically for any initial state and
any sequence of input symbols.

In the second stage we prove that the view of an adversary in $\Pi$
with {\em any} efficiently constructible corruption timeline $\rho$
and any efficiently constructible input stream is computationally
indistinguishable from the adversary's view in $I(\Pi, \rho)$. We
deduce that $\Pi$ is computationally private in the PCM.

\subsection{Constructing $I(\Pi, \rho)$}

\begin{notation}
Let $\nnPi$ denote the scheme of Section \ref{sec-from-all} that
requires all the agents for reconstruction, and let $\tnPiNaive$ 
and $\tnPi$
denote the threshold schemes of Section \ref{sec-from-t}. We say that
an adversary is {\em appropriate} for the scheme $\nnPi$ if it
corrupts at most $n-1$ agents. We say that an adversary is {\em
appropriate} for the $\tnPiNaive$ and $\tnPi$ if schemes it corrupts
at most $t$ agents.
\end{notation}

Let $\Pi$ be one of the schemes $\nnPi, \tnPiNaive$ or $\tnPi$. $\Pi$ defines initial data that an agent $A$ stores: a description of the automaton, a label for each node in the automaton and random seeds that are shared with other agents. For each scheme the domain of seeds is $\{0,1\}^{len}$ while the domain of labels is a field $\mathbb{F}$. For example, in $\nnPi$ the field is $\mathbb{F}=\mbox{GF}(2)$. The subsets of agents that share a single seed are specific to each scheme. The description of $I(\Pi, \rho)$ follows.

\smallskip
\noindent
{\bf Initialization:} An agent $A_i$ is initialized with a description of the automaton as in $\Pi$. For every subset of agents $T$ such that $A_i \in T$, if the agents in $T$ share a seed in $\Pi$ that other agents do not have then $A_i$ is initialized with $m+1$ elements, $seed^T_0$, $R^T_1,\ldots,R^T_m$. $seed^T_0$ is chosen uniformly at random from $\{0,1\}^{len}$, while $R^T_1,\ldots,R^T_m$ are chosen uniformly at random and independently from $\mathbb{F}$. $A_i$ computes the initial label of the $j$-th state (indexed from $1$ to $m$) over $\mathbb{F}$ as 
$$\ell^i_j = \sum_{T, A \in T} \alpha^T R^T_j,$$
for fixed coefficients $\alpha_T \in \mathbb{F}-\{0\}$.

The agent stores $seed^T_0$ but $R^T_1$,$\ldots$,$R^T_0$ are deleted.

\smallskip
\noindent
{\bf Processing:} Each of the schemes $\Pi$ defines data processing for every clock cycle. This processing includes computing new values for each seed and new values for each node label. 

Computing new labels and seeds in $I(\Pi, \rho)$ depends on the corruption timeline $\rho$, which is defined by a sequence $\rho=((A_1, \tau_1),\ldots,(A_t, \tau_t))$ such that the adversary corrupts agent $A_i$ at time $\tau_i$ for $i=1,\ldots,t$ and $\tau_1\leq\tau_2 \leq \ldots \leq \tau_t$. 

An agent $A$ in $I(\Pi, \rho)$ begins updating $seed^T_0$ only after the corruption of the first agent $A_i$ that holds $seed^T_0$. Therefore, at any time $\tau$, $\tau \leq \tau_i$, we have that $seed^T_{\tau} = seed^T_0$. If $\tau \geq \tau_i+1$ the agent modifies $seed^T_{\tau}$ as $\Pi$ specifies for updating a seed.
 
An agent $A_i$ in $I(\Pi, \rho)$ begins updating a state label $\ell^i_j$ only after $A_i$ is corrupted at time $\tau_i$. Therefore, when the adversary corrupts $A_i$ it obtains the original label $\ell^i_j$. For every clock cycle after $\tau_i$ the agent modifies $\ell^i_j$ as $\Pi$ specifies for updating a label.

\subsection{Privacy of $I(\Pi, \rho)$}
We show that if the corruption timeline is fixed to $\rho$ then $I(\Pi, \rho)$ is private in the information-theoretic sense. In order to do so, we introduce the following definition and lemma.
\begin{definition}
Let ${\cal A}$ be a set of agents, let $H=({\cal A}, {\cal E})$ be a hypergraph and let $\mathbb{F}$ be a finite field. We call $H$ a {\em distribution hypergraph}, if for every $T \in {\cal E}$, there is an element $R^T$ chosen uniformly from $\mathbb{F}$, such that every $A\in T$ holds $R^T$ and every $A\notin T$ has no information on $R^T$.
\end{definition}

\begin{lemma} 
\label{prop-privacy}
Let ${\cal A}$ be a set of agents and let $H=({\cal A}, {\cal E})$ be a distribution hypergraph over a finite field $\mathbb{F}$. Assume that each agent $A_i$ uses a fixed set of public elements $\{\alpha^T\}_{A_i \in T}$, such that $\forall T, \alpha^T \in \mathbb{F}-\{0\}$,  to compute a label 
$$\ell^i=\sum_{T, A_i \in T} \alpha^T R^T.$$
Assume that an adversary that corrupts an agent $A_j$ obtains both the agent's label $\ell^j$ and its random strings $R^T$ for all $T$ such that $A_j \in T$ and its . Then, the label of any uncorrupted agent $A_i$ is distributed independently of the adversary's view if and only if for any subset of agents $K$ that the adversary corrupts and for any agent $A_i, A_i \not \in K$, there exists a hyper-edge $T\in {\cal E}$, such that $A_i \in T$ and $T\cap K=\emptyset$. 
\end{lemma}

\begin {proof}
Let $A_i$ be an agent such that $A_i \not \in K$ and there exists a hyper-edge $T\in {\cal E}$, such that $A_i \in T$ and $T\cap K=\emptyset$. Since $R^T$ is chosen uniformly at random from $\mathbb{F}$, which is a field, and $\alpha^T \neq 0$ we have that $\ell^i=\sum_{T, A_i \in T} \alpha^T R^T$ is distributed uniformly at random in $\mathbb{F}$ and furthemore $\ell^i$ is independent of the adversary's view.  

Conversely, assume that there exists an agent $A_i \not \in K$ such that for every $T, A_i \in T$, there exists an agent $A$, $A \neq A_i$,  such that $A \in K \cap T$. Then, the adversary obtains $R^T$ for every $T$ such that $A_i \in T$ and can therefore compute $\ell^i$ without corrupting $A_i$.
\end {proof}

\begin{proposition}
\label{prop:priv_intermediate}
For each of the three possible schemes $\nnPi, \tnPiNaive, \tnPi$ and every corruption timeline $\rho$, if the adversary is appropriate for the scheme $\Pi$ $(\Pi \in \{\nnPi, \tnPiNaive, \tnPi\})$ then $I(\Pi, \rho)$ is private in the following sense. For  every two states $s_1, s_2 \in ST$ and for any two input streams $X_1, X_2\in \Gamma^\ast$,
$\mbox{VIEW}_{\rho}^{I(\Pi,\rho)}(X_1, s_1)=\mbox{VIEW}_{\rho}^{I(\Pi,\rho)}(X_2, s_2)$. Furthermore, all the state labels are random and independent elements in a finite field $\mathbb{F}$.
\end{proposition}
\begin{proof}
The view of an adversary $I(\Pi, \rho)$ is made up of the description of the automaton, the seeds and the labels. All of these are obtained from an agent at the moment of corruption. The description of the automaton is static. The distribution of the seeds depends only on $\rho$ and as a consequence is independent of the initial state $s$ and the input stream $X$.

Therefore, the only data elements in the adversary's view that could depend on the initial state and the input stream are the state labels. 

However, just prior to the adversary corrupting an agent $A$, since the adversary is {\em appropriate} there is a subset of uncorrupted agents $T$ such that $A \in T$ and all agents in $T$ share a seed that is not known to any agent outside $T$. In $I(\nnPi, \rho)$, an appropriate adversary corrupts a total of at most $n-1$ agents and just prior to corrupting an agent there are at least two uncorrupted agents. By the definition of $\nnPi$ this pair of agents, which we denote by $T$, shares a seed. Therefore, by the definition of $I(\nnPi, \rho)$ the two agents share the elements $seed^T_0, R^T_1,\ldots,R^T_m$, which are all random and independent of the adversary's view. Such subsets $T$ of uncorrupted agents also exist in $I(\tnPiNaive, \rho)$ and $I(\tnPi, \rho)$. 

By the construction of $I(\nnPi, \rho)$, before the corruption of $A$, the state labels of $A_i$ have their initial value $\ell^i_j = \sum_{T, A \in T} \alpha^T R^T_j$ and therefore, by Lemma \ref{prop-privacy} these labels are random and independent of the adversary's view.

Therefore, the adversary's view in $I(\Pi_i, \rho)$ with corruption timeline $\rho$ is distributed identically for any initial state $s$ and any input stream $X$.
\end{proof}

\subsection{Computational Privacy of $\nnPi, \tnPiNaive$ and $\tnPi$}

We complete the analysis by proving that $\nnPi, \tnPiNaive$ and $\tnPi$ are computationally private in the PCM. 

\begin{notation}
Let $\kappa$ be a security parameter, let $\mathbb{F}$ be a field and let $q(\kappa)$ be a polynomial. Let $m, t, n$ and $len$ be parameters such that $t^2{n-1 \choose t-1}(m\abs{\mathbb{F}}+len) < q(\kappa)$ and let $G:\{0,1\}^{len} \longrightarrow \{0,1\}^{m\abs{\mathbb{F}}+len}$ be a pseudo-random generator. Denote the uniform distribution on $\{0,1\}^{m\abs{\mathbb{F}}+len}$ by $U$. 

\end{notation}

We regard $\mbox{VIEW}_{\rho}^{\Pi}(X, s)$ as a random variable that represents the whole view of the adversary in $\Pi$ and regard $\mbox{VIEW}_{\rho[\tau]}^{\Pi}(X, s)$ as a reduction of that view to the first $\tau$ clock cycles. Similarly, $\mbox{VIEW}_{\rho[\tau]}^{I(\Pi,\rho)}(X, s)$ represents the adversary's view of the first $\tau$ clock cycles in $I(\Pi,\rho)$.

\begin{notation}
Let $\Pi$ be one of the schemes $\nnPi,\tnPiNaive$ or $\tnPi$. For every $\tau=0,\ldots,q(\kappa)$ define $H_{\tau}$ to be a hybrid scheme which is identical to $I(\Pi,\rho)$ for any clock cycle $\tau'$ such that $\tau' \leq \tau$ and is identical to $\Pi$ for any clock cycle $\tau''$ such that $\tau'' > \tau$. Define a sequence of random variables $Y_0, Y_1, \ldots,Y_{q(\kappa)}$ as follows. Select an arbitrary initial state $s$, select an input stream $X$ from $\Gamma^{q(\kappa)}$ and select a corruption timeline $\rho=((A_1,\tau_1),\ldots,(A_t,\tau_t))$ from $\rho^{q(\kappa)}$. $Y_{\tau}$ is the view of an adversary for the scheme $H_{\tau}$ given the choices of $s$, $X$ and $\rho$.
\end{notation}

It follows from the definition of the schemes $H_0,\ldots,H_{q(\kappa)}$ that $H_0$ is $\Pi$ and $H_{q(\kappa)}$ is $I(\Pi,\rho)$. Therefore, $Y_0=\mbox{VIEW}_{\rho}^{\Pi}(X, s)$ and $Y_{q(\kappa)}=\mbox{VIEW}_{\rho}^{I(\Pi,\rho)}(X, s)$. 

Note that $H_{\tau}$ is well defined for any $\tau$ since the memory contents and the inputs of $\Pi$ and $I(\Pi, \rho)$ are all in the same domain (although the distribution of the memory contents is not identical). The only difference between $\Pi$ and $I(\Pi, \rho)$ is the processing at each clock cycle.

\begin{proposition}
\label{prop:indis_pi_int}
Let $\Pi$ be one of the schemes $\nnPi, \tnPiNaive$ or $\tnPi$. If the adversary is appropriate for $\Pi$ then $\mbox{VIEW}_{\rho}^{\Pi}(X, s) \stackrel{c}{\equiv} \mbox{VIEW}_{\rho}^{I(\Pi,\rho)}(X, s)$ for  any initial state $s$, any efficiently constructible corruption timeline $\rho \in \rho^{q(\kappa)}$ and any efficiently constructible input stream $X \in \Gamma^{q(\kappa)}$. 
\end{proposition}
\begin{proof}
We assume towards a contradiction that the views of an adversary in $\Pi$ and in $I(\Pi,\rho)$ are not computationally indistinguishable.
Therefore, there exist a probabilistic, polynomial time algorithm $D$ and a polynomial $p(\cdot)$ such that 
$$ \abs{D_{\Pi} - D_{I(\Pi,\rho)}} - \frac{1}{p(\kappa)},$$
for an infinite number of values $\kappa$. $D_{\Pi}$ denotes $\mbox{Pr}[D(\mbox{VIEW}_{\rho}^{\Pi}(X, s))=1]$ and $D_{I(\Pi,\rho)}$ denotes $\mbox{Pr}[D(\mbox{VIEW}_{\rho}^{I(\Pi,\rho)}(X, s))=1]$.

We construct an algorithm $\bar{D}$ that distinguishes between $t^2 {n-1 \choose t-1}$ independent samples of $U$ and $t^2 {n-1 \choose t-1}$ independent samples of $G(seed)$ for a random $seed \in \{0,1\}^{len}$. Since $t^2 {n-1 \choose t-1}$ is at most a polynomial in $\kappa$, the algorithm $\bar{D}$ contradicts the assumption that $G$ is a pseudo-random generator, thus proving the proposition. Denote the distribution on $t^2 {n-1 \choose t-1}$ independent samples of $U$ by $U_{long}$ and denote the distribution on $t^2 {n-1 \choose t-1}$ independent samples of $G(seed)$ by $G_{long}$.

\medskip
\noindent
{\bf Description of $\bar{D}$:} the algorithm receives as input a description of the automaton, $\kappa$, $n$ and $t$. In addition, the algorithm receives as input a binary string $z$ of length $t^2 {n-1 \choose t-1} (m\abs{\mathbb{F}}+len)$ and decides whether it is chosen from $U_{long}$ or $G_{long}$ by performing the following steps.
\begin{enumerate}
\item Choose a random initial state $s$, select an input stream $X$ from $\Gamma^{q(\kappa)}$ and select a corruption timeline $\rho=((A_1,\tau_1),\ldots,(A_t,\tau_t))$ from $\rho^{q(\kappa)}$.
\item Choose a random $\tau$ in the range $1,2,\ldots,q(\kappa)$.
\item Simulate the operation of the agents $A_1,\ldots,A_t$ in the scheme $I(\Pi, \rho)$ for the first $\tau-1$ clock cycles.
\item In the $\tau$-th clock cycle all the agents that have already been corrupted, i.e. in cycles $1$ to $\tau-1$, execute $\Pi$ (which is identical to $I(\Pi, \rho)$ for a corrupted agent). For any uncorrupted player, including those that are corrupted in the $\tau$-th cycle do the following:
\begin{enumerate}
\item Update any seed that is shared with a corrupted player as specified by $\Pi$ (which is identical to the update process of $I(\Pi, \rho)$ for such seeds).
\item For any seed that is shared by set of uncorrupted agents $T$, select a fresh string of length $m\abs{\mathbb{F}}+len$ from $z$ and parse it as $B^T||S^T$ for $S^T \in \{0,1\}^{len}$ and $B^T=b_1^T,\ldots,b_m^T$. Replace the previous seed with $S^T$. 
\item Recall that in every $\Pi$ the label of the $j$-th state, $j=1,\ldots,m$, is updated by a linear combination of previous state labels and of elements $b_j^T$ derived from expanded seeds. $\bar{D}$ updates the label in a similar way, except that for every $T$ such that $S^T$ is shared by uncorrupted agents, $b_j^T$ is selected from $z$ as described in the previous step instead of being selected from an expanded seed.  
\end{enumerate}
\item Simulate the operation of the agents $A_1,\ldots,A_t$ in the scheme $\Pi$ for the last $q(\kappa)-\tau$ clock cycles.
\item Throughout the simulation of the agents simulate the actions of an adversary with corruption timeline $\rho$.
\item Run $D$ on the adversary's view and return the result of $D$.
\end{enumerate}

We argue that if $z$ is chosen from $U_{long}$ then the view of the adversary that $\bar{D}$ simulates is $Y_{\tau}$, while if $z$ is chosen from $G_{long}$ then the view of that adversary is $Y_{\tau-1}$. Obviously, the view that the adversary obtains in the first $\tau-1$ clock cycles is identical to the view in $I(\Pi, \rho)$ and the view in the last $q(\kappa)-\tau$ clock cycles is identical to the view in $\Pi$. Therefore, we need to prove that the view in the $\tau$-th clock cycle is identical to $I(\Pi, \rho)$ if $z$ is uniformly random and identical to $\Pi$ if $z$ is selected from $G_{long}$.

$I(\Pi, \rho)$ specifies identical processing to $\Pi$ for corrupted agents and seeds shared by corrupted agents. Therefore, the  differences are in seeds that are shared only by uncorrupted agents and in state labels of uncorrupted agents. 

In the $\tau$-th clock cycle, $\bar{D}$ replaces seeds that are shared by uncorrupted agents with strings selected from $z$. If $z$ is uniformly random then these seeds are uniformly random. Therefore, in this case the distribution of the seeds is identical to the distribution if $I(\Pi, \rho)$ is executed in the previous clock cycle, $\tau-1$. If $z$ is a sequence of elements of the form $G(seed)$, where $seed$ is random, then the new seed, $S^T$ is exactly as specified by $\Pi$ after a single clock tick. That is the expected distribution if the agents run $\Pi$ in the previous clock tick, $\tau-1$.

The state labels are updated by a linear combination in which the coefficients of each $b_j^T$ are non-zero. If $z$ is uniformly random then each $b_j^T$ is a random field element in $\mathbb{F}$ and therefore each state label is a a random field element in $\mathbb{F}$. By Proposition \ref{prop:priv_intermediate} that is identical to the distribution of state lables in $I(\Pi, \rho)$. If $z$ is a sequence of elements of the form $G(s)$, where $s$ is random, then the new label is exactly as specified by $\Pi$ after a single clock tick.

The argument above shows that once $\tau$ is given, $\bar{D}$ distinguishes between a sequence of uniform elements and a sequence of pseudo-random elements with the same probability that $D$ distinguishes between $Y_{\tau}$ and $Y_{\tau-1}$. Since $\tau$ is chosen randomly in the range $1,2,\ldots,q(\kappa)$ and since $Y_0=\mbox{VIEW}_{\rho}^{\Pi}(X, s)$ and $Y_{q(\kappa)}=\mbox{VIEW}_{\rho}^{I(\Pi,\rho)}(X,s)$ we have that 
$$\mbox{Pr}[\bar{D}(G_{long}=1)]=\frac{1}{q(\kappa)} \sum_{\tau=1}^{q(\kappa)} \mbox{Pr}[D(Y_{\tau})=1],$$
and
$$\mbox{Pr}[\bar{D}(U_{long})=1]=\frac{1}{q(\kappa)} \sum_{\tau=1}^{q(\kappa)} \mbox{Pr}[D(Y_{\tau-1})=1].$$

Therefore,
\begin{flalign*}
& \abs{\mbox{Pr}[\bar{D}(G_{long}=1)] - \mbox{Pr}[\bar{D}(U_{long})=1]} & = \\
& \frac{1}{q(\kappa)} \abs{\mbox{Pr}[D(Y_{q(\kappa)})=1] - \mbox{Pr}[D(U)=1]} & > \\
& \frac{1}{q(\kappa)\cdot p(\kappa)}
\end{flalign*}

for an infinite number of values $\kappa$. Since $\bar{D}$ distinguishes between $U_{long}$ and $G_{long}$ we deduce that $G$ is not a pseudo-random generator and have thus reached a contradiction.
\end{proof}

\begin{theorem}
If the adversary is appropriate then the schemes $\nnPi, \tnPiNaive$ and $\tnPi$ are all computationally private in the PCM. 
\end{theorem}
\begin{proof}
By proposition \ref{prop:indis_pi_int} if the adversary is appropriate then for every efficiently constructible corruption timeline $\rho$, $\rho \in \rho^{q(\kappa)}$, every two initial states $s_1, s_2$ and every two efficiently constructible input streams $X_1, X_2$, such that $X_1, X_2 \in \Gamma^{q(\kappa)}$ we have $\mbox{VIEW}_{\rho}^{\Pi}(X_1, s_1) \stackrel{c}{\equiv} \mbox{VIEW}_{\rho}^{I(\Pi,\rho)}(X_1, s_1)$ and $\mbox{VIEW}_{\rho}^{\Pi}(X_2, s_2) \stackrel{c}{\equiv} \mbox{VIEW}_{\rho}^{I(\Pi,\rho)}(X_2, s_2)$.

By Proposition \ref{prop:priv_intermediate} we know that
$\mbox{VIEW}_{\rho}^{I(\Pi,\rho)}(X_1,
s_1)=\mbox{VIEW}_{\rho}^{I(\Pi,\rho)}(X_2, s_2)$. Therefore,
$$\mbox{VIEW}_{\rho}^{\Pi}(X_1, s_1) \stackrel{c}{\equiv}
\mbox{VIEW}_{\rho}^{\Pi}(X_2, s_2).$$
\end{proof}

\ignore{
\subsection{Minimizing the number of the seeds}
\todo{VLAD: this sect. needs to be reworked}
In the first scheme which was described in \ref{sec-from-all}, as the seed distribution we used a clique on $n$ vertices $K_n$. In total we needed ${n \choose 2}=\frac {n(n-1)}{2}$ seeds. 
An interesting question is, are we able to minimize our number of seeds. By proposition \ref{prop-privacy}, in the first scheme, if we know that our adversary can catch at most $t$ agents, then a $t$-regular graph would be enough and will be necessary. Such a graph will have $O(tn)$ edges.
What can we say about the second scheme?
}


\appendix

\end{document}

%% file: introduction-3.tex
\section{Introduction}

Distributed computing has become an integral part of a variety of
systems, including cloud computing and ``swarm'' computing, where $n$
agents perform a computation
on common inputs.  In these emerging computing paradigms, security
(i.e., privacy and correctness) of the computation is of a primary
concern.  Indeed, in
swarm computing, often considered in military contexts (e.g., unmanned
aerial vehicle (UAV) operation), security 
of the data and program state is of paramount importance;
similarly, one persistent challenge in the field of cloud
computing is ensuring the privacy of users' data, demanded by
government, commercial, and even individual cloud users.

In this work, we revisit the notion of {\em never-ending}
private distributed computation, first considered by Dolev, Garay, Gilboa and
Kolesnikov~\cite{DGGK11}. In such a computation, an unbounded sequence
of commands (or inputs) are interpreted by several machines (agents) in 
a way that no information about the inputs as well as the state of
the computation is revealed to an adversary who is able to ``corrupt''
the agents and examine their internal state, as long as up to a predetermined
threshold of the machines are corrupted.

Dolev {\em et al.}
were able to provide very strong (unconditional, or information-theoretic)
security for computations performed by a finite-state machine (FSA),
at the price however of the computation being efficient only for a small
set of functions, as in general the complexity of the computation
is exponential in the size (number of states) of the FSA computing
the function.

In this work, we {\em minimally}\footnote{Indeed, the existence
of one-way functions is considered a minimal assumption in contemporary
cryptography. In particular, we do not allow the use of public-key 
cryptography.}
weaken the original model by
additionally assuming the existence of one-way functions (and hence
consider polynomial-time adversaries---in the security parameter; more
details below), and in return achieve very high efficiency as a
function of the size of the FSA.  We stress that we still consider
computation on {\em a priori} unbounded number of inputs, and
where the online (input-processing) phase incurs {\em no communication}.
We now describe the model in more detail.

\vspace{-.1in}
\paragraph{The setting.}
As in~\cite{DGGK11}, we consider a distributed computation setting in
which a party, whom we refer to as {\it the dealer}, has a finite
state automaton (FSA) ${\cal A}$ which accepts an ({\em a priori}
unbounded) stream of inputs $x_1, x_2, \ldots$ received from an
external source. The dealer delegates the computation to agents
$A_1,\ldots, A_n$, by furnishing them with an implementation of \A.
The agents receive, in a synchronized manner, all
the inputs for ${\cal A}$ during the online input-processing phase,
where no communication whatsoever is allowed.  Finally, given a signal
from the dealer, the agents terminate the execution, submit their 
 internal state to the dealer, who computes the  state of ${\cal A}$  and returns it as output.

We consider an attack model where an entitiy,
called the adversary, \Ad, is able to adaptively ``corrupt'' agents
(i.e., inspect their internal state) during the online execution
phase, up to a threshold\footnote{We note that more general access 
structures may be naturally employed with our constructions.} 
$t < n$.  We do not aim at maintaining the privacy of the
automaton \A; however, we wish to protect the secrecy of the state of
\A\ and the inputs' history.  We note that \Ad\ may have external
information about the computation, such as partial inputs or length of
the input sequence, state information, etc.  This auxiliary
information, together with the knowledge of \A, may exclude the
protection of certain configurations, or even fully determine \A's
state.  We stress that this cannot be avoided in any implementation,
and we do not consider this an insecurity.  Thus, our goal is to
prevent the leakage or derivation by
\Ad\ of any knowledge from seeing the execution traces which \Ad\ did
not already possess.

As mentioned above, our constructions relying on one-way functions
dictates that the computational power of entities (adversary, agents),
be polynomially bounded (in
$\kappa$, the security parameter).
Similarly, our protocols run on input streams of polynomial length.
At the same time, we do not impose an {\it a priori} bound on its length; 
moreover, the size of the agents' state is independent of it.
This allows  to use agents of the same (small) complexity (storage and 
computational power) in all situations.

\vspace{-.1in}

\paragraph{Our contributions.}
Our work is the first significant extension of the work
of~\cite{DGGK11}.  Towards our goal of making never-ending and private
distributed computation practical, we introduce an additional (minimal)
assumption of existence of one-way functions (and hence pseudo-random
number generators [PRGs]), and propose the following constructions:
\begin{tiret}
\item  A scheme with $(n,n)$ reconstruction (where all $n$ agents participate in reconstruction), where
the storage and processing time per input symbol is $O(mn)$ for each
agent. The reconstruction complexity is $O(mn$).
\item A scheme with $(t+1,n)$ reconstruction (where $t$ corrupted agents do not 
take part in the reconstruction), where the above costs are 
$O (m{n-1 \choose t-1})$.\footnote{For some values of $t$, e.g. $t=\frac{n}{2}$, this quantity would be exponential in $n$. This does not contradict our
assumption on the computational power of the participants; it simply
means that, given $\kappa$, for some values of $n$ and $t$ this
protocol cannot be executed in the allowed time.}
\end{tiret} 
Regarding tools and techniques, the carefully orchestrated use of PRGs
and secret-sharing techniques~\cite{Sha79} allows our protocols to hide the state of the computation against an adaptive adversary by using share re-randomization. 
Typically, in the context of secret sharing, this is simply done by the addition of a suitable (i.e., passing through the origin) random polynomial. However, due to the no-communication requirement, share re-randomization is a lot more challenging in our setting. This is particularly so in the more general case of the $(t+1,n)$-reconstruction protocol. We achieve share re-randomization by sharing PRG seeds among the players in a manner which allows players to achieve sufficient synchronization of their randomness, which is resilient to $t$ corruptions.

\vspace{-.1in}

\paragraph{Related work.} 
%
%
Reflecting a well-known phenomenon in distributed computing, where
a single point of failure needs to be avoided,
a team of agents (e.g., UAVs)  that collaborate in a mission
is more robust than a single agent trying
to complete a mission by itself (e.g., \cite{Bamberger,BDDS10}). 
Several techniques have been  suggested for this purpose;
another related line of work is that of automaton splitting and
replication, yielding
designs that can tolerate faults and as well as provide
some form of privacy of the computation
(see, e.g.,~\cite{DKS07,DLY07,DGKPS10,DGGK09,DGGK11}).
As mentioned above, 
only~\cite{DGGK11}
addresses the unbounded-input-stream scenario.
\ignore{
Techniques for information-theoretic communication among the participants of 
such systems have been investigated as well~\cite{DO00,BD03,DGGN08,DGKPS10},
some of which might be amenable for enhancement and application to our
scenario as well.\todo{VLAD: which techniques can be amenable? We have unbounded-input non-interactive online phase}
} 


Recall that in {\em secure multi-party
computation}~\cite{DBLP:conf/stoc/GoldreichMW87,BGW88,CCD88}, $n$
parties, some of which might be corrupted, are to compute an $n$-ary
(public) function on their inputs, in such a way that no information
is revealed about them beyond what is revealed by the function's
output.  At a high level, we similarly aim in our context to ensure
the correctness and privacy of the distributed computation. However,
as explained in~\cite{DGGK11}, our setting is significantly different
from that of MPC, and MPC definitions and solutions cannot be directly
applied here.  The reason is two-fold: MPC protects players {\em
individual} inputs, whereas in our setting the inputs are common to
all player.  Secondly, and more importantly, MPC operates on inputs of
fixed length, which would require an {\it a priori} estimate on the 
maximum input size $s_{max}$ and agents' storage linear in $s_{max}$.  
While unbounded inputs could be processed, 
by for example processing them ``in blocks,'' this would require
communication during the online phase, which is not allowed
in our setting. Refer to~\cite{DGGK11} for a more detailed discussion
on the unbounded inputs setting {\em vis-\`{a}-vis} MPC's.


Finally, we note that using recently proposed fully-homomorphic encryption
(FHE)~\cite{DBLP:conf/stoc/Gentry09} (and follow-ups) trivially solves
the problem we pose, as under FHE the agents can simply compute arbitrary
functions. In fact, plain additively homomorphic encryption 
(e.g.,~\cite{EC:Paillier98}) can
be used to encrypt the current state of the FSA and non-interactively
update it as computation progresses, in a manner similar to what is
described in our constructions (see the high-level intuition in
Section~\ref{outline}).  We note that, firstly, public-key encryption
and, dramatically so, FHE, suffer from orders-of-magnitude
computational overhead, as compared to the symmetric-key operations that
we rely on.  More importantly, in this work we aim at minimizing
the assumptions needed for efficient unbounded private distributed
computation.

\vspace{-.1in}

\paragraph{Organization of the paper.} The remainder of the paper is
organized as follows. In Section~\ref{sec-Prelim_and_outline} we present
in more detail the model, definitions and building blocks that we use 
throughout the paper. We dedicate Section~\ref{outline} to a high-level
description of our constructions, while in Section \ref{sec-from-all} we 
present them in detail. The full privacy analysis is presented in Section \ref{sec-priv-ana}. 
\ignore{The first scheme $\Pi_1$, in which we need all of the agents to reconstruct the current state of the computation is described in section~\ref{sec-from-all-details}. In section \ref{sec-from-t} we present two schemes $\Pi_2$ and $\Pi_3$, in which we need only a subset of the agents to reconstruct the current state of the computation. Sections \ref{sec-from-all-details} and \ref{sec-from-t} contain proofs of correctness and security.
} 

%% file: model.tex
\section{Model and Definitions}
\label{sec-Prelim_and_outline}

A {\it finite-state automaton} (FSA) \A\ has a finite set of states
$ST$, a finite alphabet $\Sigma$, and a transition
function $\mu: ST \times \Sigma \longrightarrow ST$. In this work we
do not assume an initial state or a terminal state for the automaton,
i.e., it may begin its execution from any state and does not
necessarily stop.

We already described in the previous section the distributed
computation setting---dealer, agents, adversary, and unbounded input
stream---under which the FSA is to be executed. In more detail, we
assume a {\em global clock} to which all agents are synchronized.  We
will assume that no more than one input symbol arrives during any
clock tick.  By {\em input stream}, 
we mean a sequence of input symbols arriving at a certain schedule of
clock ticks.  Abusing notation, we will sometimes refer to the input
without explicit reference to the schedule.  (We note that the global
clock requirement can in principle be removed if we allow the input
schedule to be leaked to \Ad.)

We also
mentioned that \Ad\ is
allowed to corrupt agents as the execution of the protocol
proceeds. We consider the so-called {\em passive} or {\em semi-honest}
adversary model, where corrupted agents can combine their views in
order to learn protected information, but are not allowed to deviate
from the protocol. Furthermore, each agent can be corrupted only once
during an execution. When it does, \Ad\ can view the entire contents
of a corrupted agent's memory, but does not obtain any of the global
inputs. 

Incidentally, we consider event processing by an agent as an
{\em atomic operation}.  That is, agents cannot be corrupted during an
execution of
state update.  This
is a natural and easily achievable assumption, which allows us to not
worry about some tedious details.  The computation is then considered
to be secure, if any two executions (possibly on different inputs and
initial states---defined more formally
below) 
are ``similarly''
distributed.  

This model of security for distributed computation on unbounded input
streams was introduced by Dolev {\em et al.}~\cite{DGGK11} as the {\it
progressive corruption} model (PCM), allowing \Ad\ to be
computationally unbounded, and in particular requiring that the
distributions of the two executions (again, more formally defined
below) be identical.

In this work we use a variant of PCM, applying the following two
weakenings to the PCM definition:
\begin{newenum}
\item Rather than requiring that the distributions of executions be identical,
we require them to be {\em computationally} indistinguishable.
This means that we guarantee security only against polynomial-time-bounded
adversaries.

\item We 
require indistinguishability of executions for the {\em same} corruption
timeline (and, of course, different input streams).  This means that,
for example, agent IDs are now allowed to be included in the agents'
views. (We use agent IDs in one of our contructions.)
We stress that this is not a significant security weakening,
as essentially we only allow the adversary to differentiate among the
agents' identities;
the inputs and current state of the computation remain computationally
hidden.
\end{newenum}

We now present our amended PCM definition.  We first formalize the
notion of 
{\em corruption timeline} and  the view of the adversary.

\begin{definition}
A {\em corruption timeline} $\rho$ is a sequence $\rho=((A_1,
\tau_1),\ldots,(A_k, \tau_k))$, where $A_1,\ldots,A_k$ are the
corrupted agents and $\tau_1,\ldots,\tau_k$ ($\tau_1 \leq \ldots \leq \tau_k$)
 denote the time when the corresponding corruption took place.
The {\em length} of
a corruption timeline is $|\rho|=k$.
\end{definition} 



We denote by $\view_{\rho}^{\rmPi}(X, s)$ the probability
distribution of the aggregated internal states of corrupted agents at
the time of corruption, when executed on input $X$ and
initial state $s$.

\begin{definition} [\bf Computational Privacy in the Progressive Corruption Model]
\label{defProgressive}
We say that a distributed computation scheme $\rmPi$ is {\rm
$t$-private in the Progressive Corruption Model (PCM)} if for every
two states $s_1, s_2 \in ST$, polynomial-length input streams $X_1,
X_2$, and any corruption timeline $\rho$, $|\rho| \leq
t$,
\[
\view_{\rho}^{\rmPi}(X_1, s_1) \compind \view_{\rho}^{\rmPi}(X_2, s_2).
\]
\end{definition}
\noindent Here, `$\compind$' denotes the computational indistinguishability
of two distributions.



\subsection{Tools and Building Blocks}
\label{sec:tools-app}
A pseudo-random generator (PRG) $G : X \rightarrow Y$, where $X$ and $Y$ are typically of the
form $\{0,1\}^k$ and $\{0,1\}^{k+l}$, respectively, for some positive
integers $k,l$. Recall that PRGs are known to exist based on the
existence of one-way functions, and that the security property of a
PRG guarantees that it is computationally infeasible to distinguish
its output on a value chosen uniformly at random from $X$ from a value
chosen uniformly at random from $Y$ (see,
e.g.,~\cite{Goldreich2000}). In our setting, we will further assume that
the old values
of the PRG seeds are securely erased by the agents upon use and hence are
not included in the view of the adversary.

The other basic tool that our protocols make use of 
is
{\em secret sharing}~\cite{Sha79}, where
essentially, a secret piece of information is ``split'' into shares
and handed out to a set of players by a distinguished player called
{\em the dealer}, in such a way that up to a threshold $t < n$ of the
players pulling together their shares are not able to learn anything
about it, while $t + 1$ are able to reconstruct the secret. We
present the specific instantiations of secret sharing as needed
in the 
corresponding sections.

%% file: overview.tex
\section{Overview of Our Approach}
\label{outline}
Let $\cal A$ be a publicly known automaton with $m$ states. We assume
that we have some ordering of the states of \A, which are denoted
by corresponding labels.  Every agent
stores the description of the automaton. In addition, during the
computation, for every state $s_j$ of \A, every agent $A_i$ computes
and stores its current label $\ell _j ^ i$. 
As mentioned above, all agents receive a global input stream
$\Gamma=\gamma_1,\gamma_2,...\gamma_i,...$ and perform computation in
synchronized time steps.

At a high level, the main idea behind our constructions is that the
state labels will be shares ({\em \`{a} la} secret
sharing~\cite{Sha79}) of a secret which identifies the currently
active state of $\cal A$.  More specifically, for each of the $m$
automaton states, the $n$ state labels (held by the $n$ agents) will
be shares of a $1$ if the state is currently active, and shares of a
$0$ otherwise.  We will show how the players' local computation on
their shares will ensure that this property is maintained throughout
the computation on the entire input stream $\Gamma$.  When the input
stream $\Gamma$ is fully processed (or a stop signal is issued), the
agents recover the current state by reconstructing the secrets
corresponding to each automaton state.  At the same time, shares of
the secrets (when not taken all together) reveal no information on the
current state of $\cal A$.

We now present additional high-level details on two variants of
the approach above. Recall that we consider the semi-honest adversary model, 
where corrupted players are not allowed to deviate from the protocol, but 
combine their views in order to learn protected information.

\vspace{-.1in}
\paragraph{$(n,n)$-reconstruction.}  In this scenario, we require that all 
$n$ agents participate in the reconstruction of the secret (corrupted
players are considered semi-honest and hence honestly provide their computed
shares).

At the onset of computation, the shares are initialized using an $(n,n)$
additive secret-sharing scheme, such that the initial state labels are the
sharing of $1$, and labels of each of the other states are shares of $0$.
When processing a global input symbol $\gamma$, each agent computes a
new label for a state $s$ by summing the previous labels of all states
$s'$ such that $\mu(s',\gamma)=s$.  It is easy to see that, due to the
fact that we use additive secret sharing, the newly computed shares
will maintain the desired
secret-sharing property.  Indeed, say that on input symbol
$\gamma$, $u$ states transition into state $s$.  If all of them
were inactive and their labels were shares of $0$'s, then the newly computed 
shares 
will encode a $0$ (as the sum of $u$ zeros).  Similarly, if one of the $u$
predecessor states was active and and its label shared a $1$, then  
the new active state $s$ will also correspond to a share a $1$.

A technical problem arises in the case of ``empty'' states, i.e.,
those that do not have incoming transitions for symbol $\gamma$, and
hence their labels are undefined. Indeed, to hide the state of the
automaton from the adversary who corrupts agent(s), we need to ensure that
each label is a random share of the appropriate secret.  Hence, we
need to generate a random $0$-share for each empty state without
communication among the agents.




In the $(n,n)$ sharing and reconstruction scenario, we will non-interactively
generate these labels pseudo-randomly as follows.  Each pair of
agents $(A_i,A_j)$ will be assigned a random PRG seed $\seed_{ij}$ 
Then, at each event (e.g., processing input symbol
$\gamma$), each agent $A_i$ will pseudo-randomly generate a string
$r_j$ using each of the seeds $\seed_{ij}$, and set the label of the
empty state to be the sum of all strings $r_j$. 
This is done for each empty state independently.  The PRG
seeds are then (deterministically) ``evolved'' thereby erasing from the
agent's view the knowledge of the labels' provenance, and making them all
indistinguishable from random.  As all agents are synchronized with
respect to the input and the shared seeds, it is easy to see that the
shares generated this way reconstruct a $0$, since each string $r_j$ will be
included twice in the total sum, and hence will cancel out (we will use an
appropriate [e.g., XOR-based] secret-sharing scheme such that 
this is ensured.).

Finally, and intuitively, we observe that  PCM security
will hold since
the view of each corrupted agent only includes pseudo-randomly
generated labels for each state and the current PRG seed value.  As noted
above, even when combined with the views of other corrupted players, the labels
are still indistinguishable from random.

\vspace{-.1in}
\paragraph{$(t+1, n)$-reconstruction.} In this scenario, up to $t$ corrupted 
agents do not take part in the reconstruction (this is motivated by the
possibility of agents (UAVs) being captured or destroyed by the
adversary).  Agents who submit their inputs are doing so correctly.
Thus, here we require $n>2t$.

We will take our $(n, n)$-reconstruction solution as the basis, and
adapt and expand it as follows.  First, in order to enable
reconstruction with $t+1$ agents, we will use $(t+1,n)$ additive
secret-sharing (such as Shamir's~\cite{Sha79}).  Second, as before,
we will use a PRG to generate labels, but now we will have a
separate seed for each subset of agents of size $n-t+1$.  Then, at each
event (e.g., processing of an input symbol),
\ignore{
 $\gamma$
\textcolor{blue}{\todo{Lena: clock cycle with or without input letter}
VLAD: We need to agree on the detail how we schedule inputs. I think
it should be at most one input per clock cycle, right.  We should
write and justify it in the model section above.},
}
each agent $A_i$,
for each of the groups he belongs to, will update its shares by
generating a random $(t+1,n)$-secret sharing of a $0$ using the randomness generated by applying $G$
to the group's seed.  Then, agent $A_i$ will use the share thus generated for
the $i$-th agent as its own,
and set the label of the empty state to be the sum of all such shares.

Here we note that, since agents are excluded from some of the groups,
and that in this scenario up to $t$ agents might not return their
state during reconstruction, special care must be taken in the
generation of the re-randomizing polynomials so that all agents have
invariantly consistent shares, even for groups they do not belong to,
and that any set of agents of size $t+1$ enable the reconstruction of
the secrets.  (See Section~\ref{sec-from-t} for details.)
The above is done for each empty state independently.
As before, the PRG seeds are then (deterministically) evolved,
making them all indistinguishable from random.

\ignore{

{\bf $(t+1, n)$-reconstruction using additively homomorphic encryption.}  In our third scenario, we will make use of additively homomorphic public-key encryption, such as Goldwasser-Micali~\cite{C:GolMic88} (for $(n,n)$ reconstruction) or Paillier~\cite{EC:Paillier98}.

Our main idea here is to encrypt the state labels with the public key known to the dealer (or reconstructible by a threshold of players).  The output is then obtained by running a secure multiparty computation on the player's current states.   Now, during FSA computation, since the labels are encrypted, they hide the state update history, and empty states can be simply labeled with encryptions of $0$ (for $(n,n)$ reconstruction) and with Shamir shares of $0$ derived from a single seed.

} 

\begin{algorithm*} [htb]
\caption{Template algorithm for agent $A_i$, $1\leq i \leq n$, for label 
and state update.\label{alg:updating}}
\label{algo:calc}
\begin{algorithmic}[1]
\REQUIRE An input symbol $\gamma$.
\ENSURE New labels for every state.
\IF{$\gamma$ is initialized} 
\STATE $ \ell ^i_j := \sum_{k, \mu(s_k, \gamma)=s_j} \ell^i_k$ (the sum is calculated over some field $\mathbb{F}$, depending on the scheme).
\ENDIF
\FOR {every $T\in {\cal T}$ s.t. $A_i\in T$}
\STATE Compute  $B^T S^T \leftarrow G(\seed^T_r)$, where $B^T=b^T_1b^T_2...b^T_m$, and $b^T_j \in \mathbb {F}$, $1\le j \le m$.
\STATE $\seed^T_{r+1} := S^T$.
\FOR {$j=1$ to $m$}
\STATE $\ell_j^i := \ell_j^i+ R_j$, where $R_j$ is a scheme-specific pseudo-random quantity.
\ENDFOR
\ENDFOR
\end{algorithmic}
\end{algorithm*}

\begin{remark}
{\em This approach reveals the length and schedule of the input
$\Gamma$ processed by the players.  Indeed, the stored seeds (or more
precisely, their evolution which is traceable by the adversary simply
by corrupting at different times players who share a seed) do reveal
to the adversary the number of times the update function has been invoked.
We hide this information by
by requiring the agents to run updates at each
clock tick.
} 
\end{remark}


Algorithm~\ref{alg:updating} 
summarizes the update operations performed by agent $A_i$ ($1\le i \le
n$) during the $r$-th clock cycle. The key point is the generation of
$R_j$, the label re-randomizing quantity. Notice also that in every
clock cycle, there may or may not be an input symbol received by the
agent; if the agent did not receive any input, we assume that the
input symbol is not initialized.


\ignore{

\todo{VLAD: I haven't gone through the algorithm. I suspect it needs some care.}

\todo{need to add text explaining that add-randomness does, depending on the instance of the protocol... what does it do?}

\todo{add-randomness is not a good name}

\todo{Juan: I tried to simplify the notation, and made add-randomness explicit. I'll probably change it further once I go over the next section.}

} 